\documentclass[sn-basic]{sn-jnl}

\usepackage{graphicx}%
\usepackage{multirow}%
\usepackage{amsmath,amssymb,amsfonts}%
\usepackage{amsthm}%
\usepackage{mathrsfs}%
\usepackage[title]{appendix}%
\usepackage{xcolor}%
\usepackage{textcomp}%
\usepackage{manyfoot}%
\usepackage{booktabs}%
\usepackage{algorithm}%
\usepackage{algorithmicx}%
\usepackage{algpseudocode}%
\usepackage{listings}%

\usepackage{pgfplots}
\pgfplotsset{width=10cm,compat=1.15}

\theoremstyle{thmstyleone}%
\newtheorem{theorem}{Theorem}
\newtheorem{lemma}[theorem]{Lemma}

\theoremstyle{thmstyletwo}%

\theoremstyle{thmstylethree}%

\begin{document}
	
	\title[Improved algorithms for single machine serial-batch scheduling]{Improved algorithms for single machine serial-batch scheduling to minimize makespan and maximum cost}
	
	
	\author*[1]{\fnm{Shuguang} \sur{Li}}\email{sgliytu@hotmail.com}
	
	\author[1]{\fnm{Zhenxin} \sur{Wen}}\email{2023420176@sdtbu.edu.cn}
	
	\author[1]{\fnm{Jing} \sur{Wei}}\email{2023420185@sdtbu.edu.cn}
	
	\affil[1]{\orgdiv{School of Computer Science and Technology}, \orgname{Shandong Technology and Business University}, \orgaddress{\city{Yantai}, \postcode{264005}, \country{China}}}

	
	\abstract{This paper studies the bicriteria problem of scheduling $n$ jobs on a serial-batch machine to minimize makespan and maximum cost simultaneously. A serial-batch machine can process up to $b$ jobs as a batch, where $b$ is known as the batch capacity. When a new batch starts, a constant setup time is required for the machine. Within each batch, the jobs are processed sequentially, and thus the processing time of a batch equals the sum of the processing times of its jobs. All the jobs in a batch have the same completion time, namely, the completion time of the batch. The main result is an $O(n^3)$-time algorithm which can generate all Pareto optimal points for the bounded model ($b<n$) without precedence relation. The algorithm can be modified to solve the unbounded model ($b\ge n$) with strict precedence relation in $O(n^3)$ time as well. The results improve the previously best known running time of $O(n^4)$ for both the bounded and unbounded models.}

	
	
	
	\maketitle
	
	\section{Introduction}\label{sec1}
	
	Let a set of $n$ jobs $\mathcal{J}=\{{{J}_{1}},{{J}_{2}},\ldots ,{{J}_{n}}\}$ and a serial-batch machine be given.  Each job ${{J}_{j}}$ becomes available for processing at time zero, and has a processing time ${{p}_{j}}$ and a due date ${{d}_{j}}$ which indicates a preferred completion time. The jobs are processed in batches on the serial-batch machine. When a new batch starts, a constant setup time $s$ is required for the machine. Jobs within each batch are processed sequentially so that the processing time of a batch equals the sum of the processing times of all jobs in the batch. All the jobs in a batch have the same completion time, namely, the completion time of the batch.
	
	The serial-batch machine can process up to $b$ jobs as a batch, where $b$ is known as the batch capacity. There are two models in the literature with respect to the batch capacity $b$: the bounded model and the unbounded model, for which $b<n$ and $b\ge n$, respectively.
	
	A schedule specifies a partition of the jobs into batches and the start times of the batches. The order of jobs within a batch is immaterial due to the definition of the job completion times. For a given schedule $\sigma $, let ${{C}_{j}}(\sigma )$ and ${{L}_{j}}(\sigma )={{C}_{j}}(\sigma )-{{d}_{j}}$ denote the completion time and lateness of job ${{J}_{j}}$ in $\sigma $, respectively. Let  ${{C}_{\max }}(\sigma )={{\max }_{j}}{{C}_{j}}(\sigma )$ and ${{L}_{\max }}(\sigma )={{\max }_{j}}{{L}_{j}}(\sigma )$ denote the maximum completion time (makespan) and maximum lateness of all jobs in $\sigma $, respectively. We omit the argument $\sigma $ when it is clear to which schedule we are referring.
	
	The maximum lateness can be generalized to maximum cost ${{f}_{\max }}(\sigma )={{\max }_{j}}{{f}_{j}}({{C}_{j}}(\sigma ))$, where ${{f}_{j}}({{C}_{j}}(\sigma ))$ is the cost associated with job ${{J}_{j}}$. It is assumed that all the cost functions are regular (non-decreasing in the completion times of the jobs), and the values of ${{f}_{j}}(t)$ for every given time $t\ge 0$ can be calculated in constant time.
	
	In addition, there may be a precedence relation imposed on the jobs. Two types of precedence relation are involved: (1) strict precedence relation ($\prec $): if ${{J}_{i}}\prec {{J}_{j}}$, then ${{C}_{i}}(\sigma )\le {{S}_{j}}(\sigma )$ in a feasible schedule $\sigma $; (2) weak precedence relation ($\preceq$): if ${{J}_{i}}\preceq {{J}_{j}}$, then ${{S}_{i}}(\sigma )\le {{S}_{j}}(\sigma )$ in a feasible schedule $\sigma $.

	In this paper, we study the problems of scheduling jobs, without or with precedence relations, on a serial-batch machine so as to simultaneously optimize makespan ${{C}_{\max }}$ and maximum cost ${{f}_{\max }}$. The problems fall under the category of Pareto optimization scheduling \citep{Hoogeveen05,TKindt06}. A feasible schedule $\sigma $ is Pareto optimal with respect to ${{C}_{\max }}$ and ${{f}_{\max }}$ if there is no feasible schedule ${\sigma }'$ such that ${{C}_{\max }}({\sigma }')\le {{C}_{\max }}(\sigma )$ and ${{f}_{\max }}({\sigma }')\le {{f}_{\max }}(\sigma )$, where at least one of the inequalities is strict. The objective vector  $({{C}_{\max }}(\sigma ),{{f}_{\max }}(\sigma ))$ is called a Pareto optimal point. Following the notation scheme of \cite{TKindt06}, the problems under consideration can be denoted as $1\vert s-batch,b<n \vert ({{C}_{\max }},{{f}_{\max }})$, $1\vert \preceq ,s-batch,b\ge n\vert ({{C}_{\max }},{{f}_{\max }})$ and $1\vert \prec ,s-batch,b\ge n\vert ({{C}_{\max }},{{f}_{\max }})$, where ``$s-batch$" stands for the serial-batch scheduling, ``$b<n$" (``$b\ge n$") means that the batch capacity is bounded (unbounded), ``$\preceq $" (``$\prec $") stands for the weak (strict) precedence relation, and the objective ``$({{C}_{\max }},{{f}_{\max }})$" represents the simultaneous optimization for criteria ${{C}_{\max }}$ and ${{f}_{\max }}$. Note that if jobs have different release dates, then both $1\vert {{r}_{j}}\vert {{L}_{\max }}$($b=1$, minimizing maximum lateness) and $1\vert s-batch,b\ge n,{{r}_{j}}\vert {{L}_{\max }}$ (without precedence relation, minimizing maximum lateness) are NP-hard problems \citep{Brucker07,Lenstra77}.

	The main contribution of this paper is an $O({{n}^{3}})$-time algorithm for $1\vert s-batch,b<n\vert ({{C}_{\max }},{{f}_{\max }})$ which can generate all Pareto optimal points and find a corresponding schedule for each Pareto optimal point. The algorithm can be modified to solve $1\vert \prec ,s-batch,b\ge n\vert ({{C}_{\max }},{{f}_{\max }})$ in  $O({{n}^{3}})$ time as well. The two problems have the previously best known running time of $O({{n}^{4}})$ \citep{Geng18}.
	
	The paper is organized as follows. In Section 2, the previous related studies are reviewed. In Section 3, an $O({{n}^{3}})$-time algorithm for $1\vert s-batch,b<n\vert ({{C}_{\max }},{{f}_{\max }})$ is presented. In Section 4, an $O({{n}^{3}})$-time algorithm for  $1\vert \prec ,s-batch,b\ge n\vert ({{C}_{\max }},{{f}_{\max }})$ is presented. Finally, some concluding remarks are provided in Section 5.
	
	\section{Related work} \label{sec2}
	
	In the last two decades, serial-batch scheduling problems have been widely investigated \citep{Geng18,Baptiste00,Cheng01,Yuan04,Ng02,Webster95,He08,He13,He15}. Along with other results, \cite{Baptiste00} presented an $O({{n}^{14}}\log n)$-time algorithm for $1\vert s-batch,b\ge n,{{r}_{j}},{{p}_{j}}=p\vert {{L}_{\max }}$ and $1\vert s-batch,b<n,{{r}_{j}},{{p}_{j}}=p\vert {{L}_{\max }}$ (jobs have different release dates and equal processing times). \cite{Cheng01} studied serial batch scheduling problems for minimizing some regular cost functions. When the jobs have equal release dates, they presented dynamic programming algorithms for minimizing the maximum lateness, the number of late jobs, the total tardiness, the total weighted completion time, and the total weighted tardiness when all due dates are equal, which are polynomial if there is a fixed number of distinct due dates or processing times. They also derived more efficient algorithms for some special cases, and proved the NP-hardness of several special cases of the bounded model.
	
	By appropriately modifying the release dates and due dates, problems $1\vert \preceq, s-batch, b\ge n,{{r}_{j}},{{p}_{j}}=p\vert {{L}_{\max }}$  (weak precedence relation) and $1\vert \preceq ,s-batch,b\ge n\vert {{L}_{\max }}$( jobs have different processing times and equal release dates) can be reduced in $O({{n}^{2}})$ time to  $1\vert s-batch,b\ge n,{{r}_{j}},{{p}_{j}}=p\vert {{L}_{\max }}$ and $1\vert s-batch,b\ge n\vert {{L}_{\max }}$  \citep{Yuan04,Ng02}, and the latter problems can be solved in $O({{n}^{14}}\log n)$ and $O({{n}^{2}})$ time, respectively \citep{Baptiste00,Webster95}.
	
	\cite{He08} presented an $O({{n}^{2}})$-time algorithm for $1\vert s-batch,b\ge n\vert ({{C}_{\max }},{{L}_{\max }})$. Later, the result has been extended to an $O({{n}^{5}})$-time algorithm for $1\vert s-batch,b\ge n\vert ({{C}_{\max }},{{f}_{\max }})$ \citep{He13}. \cite{He15} obtained an $O({{n}^{6}})$-time algorithm for $1\vert s-batch,b<n\vert ({{C}_{\max }},{{L}_{\max }})$. \cite{Geng18} presented $O({{n}^{4}})$-time algorithms for the problems studied in this paper, i.e., $1\vert s-batch,b<n\vert ({{C}_{\max }},{{f}_{\max }})$ and $1\vert \prec ,s-batch,b\ge n\vert ({{C}_{\max }},{{f}_{\max }})$. They also gave an $O({{n}^{2}})$-time algorithm for $1\vert \preceq,s-batch,b\ge n\vert ({{C}_{\max }},{{L}_{\max }})$, and proved that problems $1\vert \preceq,s-batch,b=2\vert {{L}_{\max }}$ and $1\vert \prec ,s-batch,b=2\vert {{L}_{\max }}$ are strongly NP-hard.

	\section{The bounded case}\label{sec3}
	
	In this section we will present an $O({{n}^{3}})$-time algorithm for $1\vert s-batch,b<n\vert ({{C}_{\max }},{{f}_{\max }})$. As a by-product, the last schedule constructed by the algorithm is optimal for $1\vert s-batch,b<n\vert {{f}_{\max }}$.
	
	For a batch ${{B}_{i}}$, let $p({{B}_{i}})$ denote the processing time of ${{B}_{i}}$, i.e., $p({{B}_{i}})=\sum\nolimits_{{{J}_{j}}\in {{B}_{i}}}{{{p}_{j}}}$. Let $S({{B}_{i}})$ and $C({{B}_{i}})$ denote the start time and completion time of batch ${{B}_{i}}$, respectively. We have: $C({{B}_{i}})=S({{B}_{i}})+p({{B}_{i}})$. To coincide with the definition of CSF (which will be defined later), we allow empty batches to appear in a schedule. The setup time and processing time of an empty batch are zero. Let $s({{B}_{i}})$ denote the setup time of batch ${{B}_{i}}$, which is $s$ if ${{B}_{i}}$ is nonempty, and 0 otherwise.
	
	For a feasible schedule, we index its batches left-to-right (earliest-to-latest) as ${{B}_{1}},{{B}_{2}},\ldots ,{{B}_{n}}$, where the last $l$ batches ${{B}_{n-l+1}},{{B}_{n-l+2}},\ldots ,{{B}_{n}}$ are nonempty and form a partition of $\mathcal{J}$, and the first $n-l$ batches ${{B}_{1}},{{B}_{2}},\ldots ,{{B}_{n-l}}$ are empty. Since both criteria ${{C}_{\max }}$ and ${{f}_{\max }}$ are regular, we  can consider only the schedules without idle times. Therefore, we have:

	\begin{lemma} \label{start}
		In a feasible schedule $\sigma =\left( {{B}_{1}},{{B}_{2}},\ldots ,{{B}_{n}} \right)$, $S({{B}_{1}})=s({{B}_{1}})$, $S({{B}_{i}})=C({{B}_{i-1}})+s({{B}_{i}})$, $i=2,\ldots ,n$.
		
	\end{lemma}
	
	The following definition generalizes the idea used in \citep{Lazarev17} for $1\vert {{r}_{j}},{{p}_{j}}=p\vert ({{C}_{\max }},{{L}_{\max }})$ (the classical scheduling problem where the machine can process at most one job at a time), which plays a key role in our algorithm.
	
	A family of (possibly empty) sets $F=\{{{\mathcal{J}}_{1}},{{\mathcal{J}}_{2}},\ldots ,{{\mathcal{J}}_{n}}\}$ is called a Candidate Set Family (CSF for short) if these sets are disjoint, their union is $\mathcal{J}$, and for $i=1,2,\ldots ,n$, ${{\mathcal{J}}_{i}}$ consists of the jobs which cannot be assigned to any batch whose index is larger than $i$.
	
	The definition captures the idea of maintaining implicitly a position indicator for each job in the algorithm. The initial CSF is $\{\varnothing ,\varnothing ,\ldots ,\mathcal{J}\}$, meaning that the jobs can be assigned to any batches. Let us fix a particular iteration in the algorithm, a batch and a job in this batch. We check an inequality to find whether this job is suitable for the batch. If the inequality does not hold, then we have to move this job from the batch into another one to its left and update the CSF accordingly. We will prove that during the iterations the completion time of any batch will not decrease, which implies that any job can only be moved to the left (a crucial observation in the time complexity analysis). Therefore, when a job is accommodated into component ${{\mathcal{J}}_{i}}$ of the current CSF, we know that this job cannot be assigned to any batch whose index is larger than $i$.
	
	Obviously, we require that $\sum\nolimits_{h=1}^{i}{\vert {{\mathcal{J}}_{h}} \vert }\le i\cdot b$, $i=1,2,\ldots ,n$. We say that a feasible schedule $\left( {{B}_{1}},{{B}_{2}},\ldots ,{{B}_{n}} \right)$ satisfies a given CSF $F=\{{{\mathcal{J}}_{1}},{{\mathcal{J}}_{2}},\ldots ,{{\mathcal{J}}_{n}}\}$ if (1) for $i=n,n-1,\ldots ,1$, all the jobs in batch ${{B}_{i}}$ come from $\bigcup _{h=i}^{n}{{\mathcal{J}}_{h}}$; and (2) batches ${{B}_{k}},{{B}_{k+1}},\ldots ,{{B}_{n}}$ are all nonempty, where ${{B}_{k}}$ is the first (the earliest) nonempty batch in $\left( {{B}_{1}},{{B}_{2}},\ldots ,{{B}_{n}} \right)$. Equivalently speaking, each job in ${{\mathcal{J}}_{i}}$ must be processed in a batch which is earlier than at least $n-i$ nonempty batches in the schedule, $i=n,n-1,\ldots ,1$.
	
	Let $\Pi \left( \mathcal{J} \right)$ denote the set of all feasible schedules for $\mathcal{J}$. Let $\Pi \left( \mathcal{J},F,y \right)\subseteq \Pi \left( \mathcal{J} \right)$ denote the set of the schedules which satisfy CSF $F$ and have maximum cost less than $y$. Clearly, we have $\Pi \left( \mathcal{J},{{F}^{0}},+\infty  \right)=\Pi \left( \mathcal{J} \right)$, where ${{F}^{0}}=\{\varnothing ,\varnothing ,\ldots ,\mathcal{J}\}$.

	We first give an $O({{n}^{3}}\log n)$-time algorithm for the following auxiliary problem, and then improve its time complexity to $O({{n}^{3}})$. The improved algorithm will be used as a central subroutine in the main procedure for solving $1\vert s-batch,b<n\vert ({{C}_{\max }},{{f}_{\max }})$.

	{\bf Auxiliary Problem:} Find a schedule $\sigma $ in $\Pi \left( \mathcal{J},F,y \right)$ with minimum makespan.
	
	\begin{algorithm}
		\caption{(Algorithm Aux1)}\label{<Aux1>}
		\begin{algorithmic}[0]
			\item {\bf Step 1.} Initially, set $m=0$. Let ${{F}^{m}}=\{\mathcal{J}_{1}^{m},\mathcal{J}_{2}^{m},\ldots ,\mathcal{J}_{n}^{m}\}$, where $\mathcal{J}_{i}^{m}={{\mathcal{J}}_{i}}$ and the jobs in $\mathcal{J}_{i}^{m}$ are stored in a Max-heap ordered by processing times, $i=1,2,\ldots ,n$.
			\item {\bf Step 2.} Form the batches of schedule ${{\sigma }^{m}}$: For $i=n,n-1,\ldots ,1$ (this ordering is used crucially), let batch $B_{_{i}}^{m}$ consist of the $\min \{b,\vert \bigcup _{h=i}^{n}\mathcal{J}_{h}^{m}\backslash \bigcup _{g=i+1}^{n}B_{g}^{m} \vert\}$ jobs with largest processing times in $\bigcup _{h=i}^{n}\mathcal{J}_{h}^{m}\backslash \bigcup _{g=i+1}^{n}B_{g}^{m}$. If $B_{_{i}}^{m}=\varnothing $ but there is $\mathcal{J}_{e}^{m}\ne \varnothing $ with $e<i$, then return $\varnothing $.
			\item {\bf Step 3.} Schedule the batches in ${{\sigma }^{m}}$: Let $S(B_{_{1}}^{m})=s(B_{_{1}}^{m})$, $S(B_{_{i}}^{m})=C(B_{_{i-1}}^{m})+s(B_{_{i}}^{m})$, $i=2,\ldots ,n$.
			\item {\bf Step 4.}  Adjust ${{F}^{m}}$: For $i=n,n-1,\ldots ,1$, check  the inequality ${{f}_{j}}(C(B_{_{i}}^{m}))<y$ for each job ${{J}_{j}}$ in $B_{_{i}}^{m}$. If the inequality does not hold, then find the largest index $k$ such that ${{f}_{j}}(C(B_{_{k}}^{m}))<y$. Delete ${{J}_{j}}$ from its original set in ${{F}^{m}}$ and insert it in $\mathcal{J}_{k}^{m}$. If $k$ does not exist, then return $\varnothing $.
			\item {\bf Step 5.} If no adjustment has been done for ${{F}^{m}}$ after Step 4 (i.e., for all pairs $i,j$ the inequality is correct), then return ${{\sigma }^{m}}$. Otherwise, check the condition $\sum\nolimits_{h=1}^{i}{\vert \mathcal{J}_{_{h}}^{m} \vert}\le i\cdot b$ for $i=1,2,\ldots ,n$. If for any $i$ the condition does not hold, then return $\varnothing $. Otherwise, let ${{F}^{m+1}}={{F}^{m}}$ (${{F}^{m}}$ has been adjusted already) and then set $m=m+1$. Go to the next iteration (Step 2).
		\end{algorithmic}
	\end{algorithm}
	
	In Step 1 of Algorithm AUX1, we use a Max-heap to store the jobs in $\mathcal{J}_{i}^{m}$, $i=1,2,\ldots ,n$. These jobs are ordered by their processing times.  A heap data structure is an array object which can be viewed as a nearly complete binary tree \citep{Cormen09}. The tree is completely filled on all levels except possibly the lowest, which is filled from the left up to a point. In a Max-heap, the value of any tree node is at most the value of its parent. Thus, the largest element is stored at the root. The Max-heap (with $\tau $ elements) data structure is equipped with several basic procedures: (1) MAX-HEAPIFY, running in $O(\log \tau )$ time, is the key to maintaining the Max-heap property; (2) BUILD-MAX-HEAP, running in $O(\tau )$ time, produces a Max-heap from an unordered input array; (3) MAX-HEAP-INSERT, running in $O(\log \tau )$ time, inserts an element into the Max-heap; (4) MAX-HEAP-EXTRACT, running in $O(\log \tau )$ time, gets the maximum element and delete it from the Max-heap.
	
	Step 1 of AUX1 can be implemented in $O(n)$ time (BUILD-MAX-HEAP). Step 2 can be implemented in $O(n\log n)$ time in each iteration (MAX-HEAPIFY, MAX-HEAP-INSERT and MAX-HEAP-EXTRACT). Steps 3 and 5 require $O(n)$ time in each iteration.
	
	In Step 4, each job deletion (can be accomplished via MAX-HEAPIFY) or insertion (MAX-HEAP-INSERT) requires $O(\log n)$ time. Since there are $n$ jobs and each job goes through at most $n-1$ sets (from ${{\mathcal{J}}_{n}}$ to ${{\mathcal{J}}_{1}}$), the total number of job deletion and insertion is $O(\sum\nolimits_{i}{\vert {{\mathcal{J}}_{i}} \vert})=O({{n}^{2}})$, and the total number of iterations is $O({{n}^{2}})$. Checking the inequalities for all jobs in one iteration requires $O(n)$ time. Hence, the complexity contribution of Step 4 for all iterations is $O({{n}^{3}})$.
	
	The overall running time of Algorithm AUX1 is $O({{n}^{3}}\log n)$, which is determined by Step 2 for all iterations.
	
	\begin{lemma} \label{early}
		Let ${{\sigma }^{m}}=\left( B_{1}^{m},B_{2}^{m},\ldots ,B_{n}^{m} \right)$ be the schedule obtained at iteration $m$ ($m=0,1,\ldots $) of Algorithm AUX1 subject to ${{F}^{m}}$, where the last ${{l}_{m}}$ batches $B_{n-{{l}_{m}}+1}^{m},B_{n-{{l}_{m}}+2}^{m},\ldots ,B_{n}^{m}$ are nonempty. Let $\sigma =\left( {{B}_{1}},{{B}_{2}},\ldots ,{{B}_{n}} \right)$ be any feasible schedule subject to ${{F}^{m}}$, where the last $l$ batches ${{B}_{n-l+1}},{{B}_{n-l+2}},\ldots ,{{B}_{n}}$ are nonempty. Then the following properties hold:
		
		(1) ${{l}_{m}}\le l$, i.e., ${{\sigma }^{m}}$ has the minimum number of nonempty batches among all the feasible schedules subject to ${{F}^{m}}$;
		
		(2) $S(B_{i}^{m})\le S({{B}_{i}})$, $i=1,2,\ldots ,n$;
		
		(3) $C(B_{i}^{m})\le C({{B}_{i}})$, $i=1,2,\ldots ,n$;
		
		(4) ${{C}_{\max }}({{\sigma }^{m}})\le {{C}_{\max }}(\sigma )$, i.e., ${{\sigma }^{m}}$ has the minimum makespan among all the feasible schedules subject to ${{F}^{m}}$.
		
	\end{lemma}
	
	\begin{proof}
		We will show a transformation of  $\sigma $ into ${{\sigma }^{m}}$. To prove the lemma, we ensure that the start time and completion time of any batch will not increase during the transformation.
		
		Algorithm AUX1 forms the batches of schedule ${{\sigma }^{m}}$ backwards (i.e., from higher-numbered batches to lower-numbered batches), in a greedy manner such that each batch consists of as many as possible of the available jobs with largest processing times. We will compare the batches $B_{i}^{m}$ and ${{B}_{i}}$ in schedules ${{\sigma }^{m}}$ and $\sigma $, and modify ${{B}_{i}}$ accordingly to get $B^{'}_{i}$ such that $B^{'}_{i}=B_{i}^{m}$, $i=n,n-1,\ldots ,1$.
		
		We first compare $B_{n}^{m}$ and ${{B}_{n}}$. Note that $\vert B_{n}^{m} \vert\ge \vert {{B}_{n}} \vert$. Let $A_{n}^{m}$ denote the set of the $\vert {{B}_{n}} \vert $ jobs with largest processing times in $B_{n}^{m}$. Let ${{J}_{j}}$ be one of the jobs in $A_{n}^{m}\backslash {{B}_{n}}$. In $\sigma $, ${{J}_{j}}$ is in an earlier batch than ${{B}_{n}}$. Let ${{J}_{{{j}'}}}$ be a job in ${{B}_{n}}\backslash A_{n}^{m}$. Obviously, we have: ${{p}_{j}}\ge {{p}_{{{j}'}}}$. We exchange ${{J}_{j}}$ and ${{J}_{{{j}'}}}$ in $\sigma $. Note that in $\sigma $ the start time and completion time of any batch will not increase. Repeat the operation until all the jobs in ${{B}_{n}}$ are replaced by the jobs in $A_{n}^{m}$. We then move the jobs in $B_{n}^{m}\backslash A_{n}^{m}$  from other batches in $\sigma $ into modified ${{B}_{n}}$. Thereby, we get $B^{'}_{n}=B_{n}^{m}$. The obtained schedule still satisfies ${{F}^{m}}$ because all the jobs in  ${{B}_{n}}$ and $B_{n}^{m}$ belong to ${{\mathcal{J}}_{n}}$.
		
		Now suppose that for $i=n,n-1,\ldots ,k+1$ ($1\le k\le n-1$) we have $B^{'}_{i}=B_{i}^{m}$, and no batch starts or completes later during the transformation. (At this point, although the batches ${{B}_{1}},{{B}_{2}},\cdots ,{{B}_{k}}$ may also have been changed, we will not change the notations for these batches.) Note that $\vert B_{k}^{m} \vert \ge \vert {{B}_{k}} \vert $. Let $A_{k}^{m}$ denote the set of the $\vert {{B}_{k}} \vert $ jobs with largest processing times in $B_{k}^{m}$. Let ${{J}_{j}}$ be one of the jobs in $A_{k}^{m}\backslash {{B}_{k}}$. In $\sigma $, ${{J}_{j}}$ is in an earlier batch than ${{B}_{k}}$. Let ${{J}_{{{j}'}}}$ be a job in ${{B}_{k}}\backslash A_{k}^{m}$. We have: ${{p}_{j}}\ge {{p}_{{{j}'}}}$. We exchange ${{J}_{j}}$ and ${{J}_{{{j}'}}}$ in $\sigma $. Repeat the operation until all the jobs in ${{B}_{k}}$ are replaced by the jobs in $A_{k}^{m}$. We then move the jobs in $B_{k}^{m}\backslash A_{k}^{m}$  from other batches in $\sigma $ into modified ${{B}_{k}}$. Thereby, we get ${{{B}'}_{k}}=B_{k}^{m}$. The obtained schedule still satisfies ${{F}^{m}}$ because all the jobs in  ${{B}_{k}}$ and $B_{k}^{m}$ belong to $\bigcup _{h=k}^{n}{{\mathcal{J}}_{h}}$. Since we always exchange a longer job in an earlier batch with a shorter job (possibly with processing time zero) in a later batch, the start time and completion time of any batch will not increase during the transformation. Using Lemma \ref{start}, it is easily checked that properties (1) through (4) hold.
		
	\end{proof}
	
	Lemmas \ref{equal1} and \ref{equal2} below and their proofs are adapted from \citep{Lazarev17} for the serial-batch setting.
	
	\begin{lemma}\label{equal1}
		For $m=1,2,\ldots $, Algorithm AUX1 ensures that  $\Pi \left( \mathcal{J},{{F}^{m}},y \right)=\Pi \left( \mathcal{J},{{F}^{m-1}},y \right)$.
	\end{lemma}
	
	\begin{proof}
		Obviously, we have $\Pi \left( \mathcal{J},{{F}^{m}},y \right)\subseteq \Pi \left( \mathcal{J},{{F}^{m-1}},y \right)$. Thus, it is sufficient to prove that $\Pi \left( \mathcal{J},{{F}^{m-1}},y \right)\subseteq \Pi \left( \mathcal{J},{{F}^{m}},y \right)$.
		
		Let $\sigma \in \Pi \left( \mathcal{J},{{F}^{m-1}},y \right)$. Let $\mathcal{J}_{i}^{m-1}\in {{F}^{m-1}}$ and $\mathcal{J}_{i}^{m}\in {{F}^{m}}$ with ${{J}_{j}}\in \bigcup _{h=i+1}^{n}\mathcal{J}_{h}^{m-1}$ but ${{J}_{j}}\in \mathcal{J}_{i}^{m}$. Then we have: ${{f}_{j}}(C(B_{_{i+1}}^{m-1}({{\sigma }^{m-1}})))\ge y$. By Lemma \ref{early}, we have: ${{f}_{j}}(C(B_{_{i+1}}^{m-1}(\sigma )))\ge {{f}_{j}}(C(B_{_{i+1}}^{m-1}({{\sigma }^{m-1}})))\ge y$, which means that in $\sigma $ job ${{J}_{j}}$ cannot be assigned to a batch whose index is larger than $i$, and $B_{_{i+1}}^{m-1}(\sigma ),B_{_{i+2}}^{m-1}(\sigma ),\ldots ,B_{_{n}}^{m-1}(\sigma )$ are all nonempty batches. It follows that $\sigma \in \Pi \left( \mathcal{J},{{F}^{m}},y \right)$.
		
	\end{proof}
	
	Combining Lemmas \ref{early} and \ref{equal1}, we get the following theorem, which shows that Algorithm AUX1 solves the Auxiliary Problem.
	
	\begin{theorem}\label{last}
		Let ${{\sigma }^{last}}$ be the schedule obtained at the last iteration of Algorithm AUX1. If ${{\sigma }^{last}}=\varnothing $, then $\Pi \left( \mathcal{J},F,y \right)=\varnothing $; Otherwise ${{\sigma }^{last}}$ is a schedule which has minimum makespan among all schedules in $\Pi \left( \mathcal{J},F,y \right)$.
		
	\end{theorem}

	Next, we give an improved algorithm for solving the auxiliary problem. The idea is motivated from \citep{Koehler13}, which solved a parallel batch scheduling problem on parallel machines where jobs have release dates, deadlines, and identical processing times. (In parallel batch scheduling setting, jobs within each batch are processed in parallel so that the processing time of a batch equals the largest processing time of all jobs in the batch.)
	
	\begin{algorithm}
		\caption{(Algorithm IMPROAUX1)}\label{<IMPROAUX1>}
		\begin{algorithmic}[0]
			\item {\bf Step 1.} Initially, set $m=0$. Let ${{F}^{m}}=\{\mathcal{J}_{1}^{m},\mathcal{J}_{2}^{m},\ldots ,\mathcal{J}_{n}^{m}\}$, where $\mathcal{J}_{i}^{m}={{\mathcal{J}}_{i}}$, $i=1,2,\ldots ,n$.
			\item {\bf Step 2.} The same as Step 2 of Algorithm AUX1.
			\item {\bf Step 3.} The same as Step 3 of Algorithm AUX1.
			\item {\bf Step 4.}  Adjust ${{F}^{m}}$ and ${{\sigma }^{m}}$: For $i=n,n-1,\ldots ,1$, check  the inequality ${{f}_{j}}(C(B_{_{i}}^{m}))<y$ for each job ${{J}_{j}}$ in $B_{_{i}}^{m}$. If the inequality does not hold, delete ${{J}_{j}}$ from its original set in ${{F}^{m}}$ and insert it in $\mathcal{J}_{i-1}^{m}$. Then, adjust ${{\sigma }^{m}}$ as follows. Remove ${{J}_{j}}$ from $B_{_{i}}^{m}$. Let $E(B_{_{i}}^{m})=\{{j}'\vert {j}'\in \bigcup _{g=1}^{i-1}B_{g}^{m}\wedge {{f}_{{{j}'}}}(C(B_{_{i}}^{m}))<y\}$. We distinguish two different cases:
			
			{\bf Case 1.} $E(B_{_{i}}^{m})=\varnothing $.
			
			If the modified $B_{_{i}}^{m}$ (some jobs may have been removed from $B_{_{i}}^{m}$) becomes empty, then return $\varnothing $. Otherwise, if $B_{i-1}^{m},B_{i-2}^{m},\ldots ,B_{1}^{m}$ are all full batches, then return $\varnothing $; else find the largest index $e\le i-1$ such that $\vert B_{e}^{m} \vert <b$. The adjustment of ${{\sigma }^{m}}$ moves right-to-left over consecutive batches, starting from $B_{_{i-1}}^{m}$ and ending by $B_{e}^{m}$. Let $B_{_{c}}^{m}$ denote the current batch. Let ${{J}_{r}}$ denote the current job, initially ${{J}_{j}}$. If $c>e$, let ${{J}_{z}}$ denote the job with the smallest processing time in $B_{_{c}}^{m}\cup \{{{J}_{r}}\}$. Update $B_{_{c}}^{m}$ to be $B_{_{c}}^{m}\cup \{{{J}_{r}}\}\backslash \{{{J}_{z}}\}$. Continue with batch $B_{_{c-1}}^{m}$ and job ${{J}_{z}}$. If $c=e$, simply put ${{J}_{r}}$ into  $B_{_{c}}^{m}$.
			
			{\bf Case 2.} $E(B_{_{i}}^{m})\ne \varnothing $.
			
			Pick the job in $E(B_{_{i}}^{m})$ with the largest processing time, say  ${{J}_{x}}$. Remove ${{J}_{x}}$ from its current batch $B_{_{e}}^{m}$ ($e<i$) and put it into $B_{_{i}}^{m}$. It is easy to check that $e$ is the largest index such that $e\le i-1$ and $\vert B_{e}^{m} \vert <b$. The adjustment of ${{\sigma }^{m}}$ moves right-to-left over consecutive batches, starting from $B_{_{i-1}}^{m}$ and ending by $B_{e}^{m}$, in the same way as described in Case 1.
			
			\item {\bf Step 5.} If no adjustment has been done for ${{F}^{m}}$ after Step 4, then return ${{\sigma }^{m}}$. Otherwise, let ${{F}^{m+1}}={{F}^{m}}$,  ${{\sigma }^{m+1}}={{\sigma }^{m}}$ (${{F}^{m}}$ and ${{\sigma }^{m}}$ have been adjusted already) and then set $m=m+1$. Go to Step 3.
		\end{algorithmic}
	\end{algorithm}
	
	The difference between Algorithms AUX1 and IMPROAUX1 is that IMPROAUX1 runs Step 2 only once (for forming the batches of ${{\sigma }^{0}}$). For $m>0$, ${{\sigma }^{m}}$ is constructed by adjusting ${{\sigma }^{m-1}}$. In fact, ${{\sigma }^{m}}$ can be obtained by running Step 2 for ${{F}^{m}}$ (with a higher time complexity). The correctness of IMPROAUX1 follows. Step 4 of IMPROAUX1 requires $O(n)$ time for adjusting an inequality violation. Since there are $O({{n}^{2}})$ such adjustments, the overall running time of Algorithm IMPROAUX1 is $O({{n}^{3}})$.
	
	Note that Lemma \ref{early}, Lemma \ref{equal1} and Theorem \ref{last} still hold for Algorithm IMPROAUX1.
	
	Now, we are ready to describe the main algorithm for constructing the Pareto set $\Omega (\mathcal{J})$ for $1\vert s-batch,b<n\vert ({{C}_{\max }},{{f}_{\max }})$.

	\begin{algorithm}
		\caption{(Algorithm MAIN1)}\label{<MAIN1>}
		\begin{algorithmic}[0]
			\item {\bf Step 1.} Initially, set $u=0$, ${{y}^{u}}=+\infty $ and ${{\sigma }_{u}}=\varnothing $. Let ${{F}^{u}}=\{\mathcal{J}_{1}^{u},\mathcal{J}_{2}^{u},\ldots ,\mathcal{J}_{n}^{u}\}$, where $\mathcal{J}_{n}^{u}=\mathcal{J}$ and $\mathcal{J}_{i}^{u}=\varnothing $ for $i=1,2,\ldots ,n-1$. Let $\Omega (\mathcal{J})=\varnothing $, $k=0$.
			\item {\bf Step 2.} If $u=0$, run Algorithm IMPROAUX1 to get a schedule ${{\sigma }_{u+1}}$ with minimum makespan among all schedules in $\Pi \left( \mathcal{J},{{F}^{u}},{{y}^{u}} \right)$. If $u>0$, starting with ${{\sigma }_{u}}$ from Step 3 of IMPROAUX1, run IMPROAUX1 to get a schedule ${{\sigma }_{u+1}}$ with minimum makespan among all schedules in $\Pi \left( \mathcal{J},{{F}^{u}},{{y}^{u}} \right)$.  (That is, during the entire run of Algorithm MAIN1, Step 2 of IMPROAUX1 is executed only once. For $u>0$, ${{\sigma }_{u+1}}$ is constructed by adjusting a series of tentative schedules, starting with ${{\sigma }_{u}}$ from Step 3 of IMPROAUX1.) Let upon the completion of Algorithm IMPROAUX1 ${{F}^{u+1}}=\{\mathcal{J}_{1}^{u+1},\mathcal{J}_{2}^{u+1},\ldots ,\mathcal{J}_{n}^{u+1}\}$ be obtained.
			\item {\bf Step 3.}  If ${{\sigma }_{u+1}}\ne \varnothing $, then:
			
			(i) Set ${{\pi }^{*}}(\mathcal{J})={{\sigma }_{u+1}}$ and ${{y}^{u+1}}={{f}_{\max }}({{\sigma }_{u+1}})$.
			
			(ii) If ${{C}_{\max }}({{\sigma }_{u}})<{{C}_{\max }}({{\sigma }_{u+1}})$ and $u>0$, then set $k=k+1$ and ${{\pi }_{k}}={{\sigma }_{u}}$. Include $({{C}_{\max }}({{\pi }_{k}}),{{f}_{\max }}({{\pi }_{k}}),{{\pi }_{k}})$ into $\Omega (\mathcal{J})$.
			
			(iii) Set $u=u+1$. Go to the next iteration (Step 2).
			
			\item {\bf Step 4.} If ${{\sigma }_{u+1}}=\varnothing $, then set $k=k+1$ and ${{\pi }_{k}}={{\sigma }_{u}}$. Include $({{C}_{\max }}({{\pi }_{k}}),{{f}_{\max }}({{\pi }_{k}}),{{\pi }_{k}})$ into $\Omega (\mathcal{J})$ and return $\{\Omega (\mathcal{J}),{{\pi }^{*}}(\mathcal{J})\}$.
		\end{algorithmic}
	\end{algorithm}
	
	\begin{lemma}\label{equal2}
		For $u=0,1,\ldots $, Algorithm MAIN1 ensures that  $\Pi \left( \mathcal{J},{{F}^{u}},{{y}^{u}} \right)=\Pi \left( \mathcal{J},{{F}^{0}},{{y}^{u}} \right)$, where ${{F}^{0}}=\{\varnothing ,\varnothing ,\ldots ,\mathcal{J}\}$ and ${{y}^{0}}=+\infty $.
	\end{lemma}
	
	\begin{proof}
		At the first iteration of Algorithm MAIN1, ${{\sigma }_{1}}$ and ${{F}^{1}}$ were obtained. By Lemma 3, we have: $\Pi \left( \mathcal{J},{{F}^{1}},{{y}^{0}} \right)=\Pi \left( \mathcal{J},{{F}^{0}},{{y}^{0}} \right)$. Since ${{y}^{1}}={{f}_{\max }}({{\sigma }_{1}})<{{y}^{0}}$, we have: $\Pi \left( \mathcal{J},{{F}^{1}},{{y}^{1}} \right)=\Pi \left( \mathcal{J},{{F}^{0}},{{y}^{1}} \right)$. By Lemma \ref{equal1}, we get: $\Pi \left( \mathcal{J},{{F}^{2}},{{y}^{1}} \right)=\Pi \left( \mathcal{J},{{F}^{1}},{{y}^{1}} \right)=\Pi \left( \mathcal{J},{{F}^{0}},{{y}^{1}} \right)$. Hence we get: $\Pi \left( \mathcal{J},{{F}^{2}},{{y}^{2}} \right)=\Pi \left( \mathcal{J},{{F}^{0}},{{y}^{2}} \right)$. Repeating the argument for all iterations we obtain $\Pi \left( \mathcal{J},{{F}^{u}},{{y}^{u}} \right)=\Pi \left( \mathcal{J},{{F}^{0}},{{y}^{u}} \right)$, $u=0,1,\ldots $.
		
	\end{proof}
	
	The proof of the following theorem applies a generic approach of Pareto optimization scheduling. Please refer to \citep{Hoogeveen05,He14,Geng15} for more detailed illustrations.
	
	\begin{theorem} \label{MAIN1}
		Algorithm MAIN1 solves $1\vert s-batch,b<n\vert ({{C}_{\max }},{{f}_{\max }})$ in $O({{n}^{3}})$ time. Consequently, problem $1\vert s-batch,b<n\vert {{f}_{\max }}$ can also be solved in $O({{n}^{3}})$ time.
	\end{theorem}
	
	\begin{proof}
		Note that ${{\sigma }_{i+1}}$ is the schedule with minimum makespan among all schedules in $\Pi \left( \mathcal{J},{{F}^{i}},{{y}^{i}} \right)=\Pi \left( \mathcal{J},{{F}^{0}},{{y}^{i}} \right)$, and ${{y}^{i+1}}={{f}_{\max }}({{\sigma }_{i+1}})$, $i=0,1,\ldots $.
		
		Let $\Omega (\mathcal{J})=\{({{C}_{\max }}({{\pi }_{k}}),{{f}_{\max }}({{\pi }_{k}}),{{\pi }_{k}}):k=1,2,\ldots ,l\}$ be the set returned by Algorithm MAIN1. We have: ${{C}_{\max }}({{\pi }_{1}})<{{C}_{\max }}({{\pi }_{2}})<\cdots <{{C}_{\max }}({{\pi }_{l}})$, ${{f}_{\max }}({{\pi }_{1}})>{{f}_{\max }}({{\pi }_{2}})>\cdots >{{f}_{\max }}({{\pi }_{l}})$.
		
		Let $\pi $ be a Pareto optimal schedule for $1\vert s-batch,b<n\vert ({{C}_{\max }},{{f}_{\max }})$ whose objective vector is  $({{C}_{\max }}(\pi ),{{f}_{\max }}(\pi ))$.
		
		If  ${{f}_{\max }}(\pi )\ge {{f}_{\max }}({{\pi }_{1}})$, since ${{f}_{\max }}(\pi )<+\infty $ and ${{\pi }_{1}}$ is optimal for $1\vert s-batch,b<n,{{f}_{\max }}<+\infty \vert {{C}_{\max }}$, we have ${{C}_{\max }}(\pi )\ge {{C}_{\max }}({{\pi }_{1}})$. Therefore, $({{C}_{\max }}(\pi ),{{f}_{\max }}(\pi ))$ cannot be a Pareto optimal point unless  $({{C}_{\max }}(\pi ),{{f}_{\max }}(\pi ))=({{C}_{\max }}({{\pi }_{1}}),{{f}_{\max }}({{\pi }_{1}}))$.
		
		If ${{f}_{\max }}({{\pi }_{k}})>{{f}_{\max }}(\pi )\ge {{f}_{\max }}({{\pi }_{k+1}})$ ($1\le k<l$), since ${{\pi }_{k+1}}$ is optimal for $1\vert s-batch,b<n,{{f}_{\max }}<{{f}_{\max }}({{\pi }_{k}})\vert {{C}_{\max }}$, we have ${{C}_{\max }}(\pi )\ge {{C}_{\max }}({{\pi }_{k+1}})$. Therefore, $({{C}_{\max }}(\pi ),{{f}_{\max }}(\pi ))$ cannot be a Pareto optimal point unless  $({{C}_{\max }}(\pi ),{{f}_{\max }}(\pi ))=({{C}_{\max }}({{\pi }_{k+1}}),{{f}_{\max }}({{\pi }_{k+1}}))$.
		
		If ${{f}_{\max }}(\pi )<{{f}_{\max }}({{\pi }_{l}})$, then $\pi $ does not exist.
		
		The above analysis shows that $\Omega (\mathcal{J})$ consists of all Pareto optimal points together with the corresponding schedules.
		
		At the last iteration of Algorithm MAIN1, we get ${{\sigma }_{u+1}}=\varnothing $, which means that $\Pi \left( \mathcal{J},{{F}^{u}},{{y}^{u}} \right)=\Pi \left( \mathcal{J},{{F}^{0}},{{y}^{u}} \right)=\varnothing $ (by Lemma \ref{equal2}). That is, there is no feasible schedule with maximum cost lower than ${{y}^{u}}$. Thus, ${{\pi }^{*}}(\mathcal{J})={{\sigma }_{u}}$ is optimal for single criterion ${{f}_{\max }}$. 
			\end{proof}
		
		Similarly to the analysis of time complexity of Algorithm IMPROAUX1, we know that the overall running time of Algorithm MAIN1 is $O({{n}^{3}})$.

		In this section, we present experimental results of our $1\vert s-batch,b<n\vert ({{C}_{\max }},{{f}_{\max }})$ algorithm. The algorithm, implemented in PyCharm, was tested on randomly generated instances. We varied key factors. In each experiment, we assign a different number of jobs\((J\in\{10,100\})\), randomly generate processing times and deadlines for each jobs, and provide an appropriate batch capacity based on the number of jobs. 
		The processing times uniformly distributed as integers with $p \in \{40,60\}$. The deadlines times uniformly distributed as integers with $d \in \{60,90\}$.The average and maximum running times of Algorithm  in the paper is showed in Table 1. 
		
		For this problem, a research method with a time complexity of \(O(n^{4})\) was proposed in previous studies. In this section, we compare two methods. By allocating jobs of different orders of magnitude and randomly generating the processing time of each job while keeping the range of the deadline consistent. Table 2 presents the average and maximum running times of the previous original algorithm.

		\begin{table}
			\centering
			\caption{}
			{\renewcommand{\arraystretch}{3}}
				\begin{tabular}{ccc}
					\toprule
					number & Average-time(s) & Max-time(s) \\
					\midrule
					10 & 1.2922E-04 & 2.3246E-04 \\
					20 & 3.1762E-04 & 7.5841E-04 \\
					30 & 5.4915E-04 & 7.5346E-04 \\
					40 & 8.6671E-04 & 1.1714E-03 \\
					50 & 1.2306E-03 & 1.2851E-03 \\
					60 & 1.6678E-03 & 2.4214E-03 \\
					70 & 2.1799E-03 & 2.2566E-03 \\
					80 & 2.7600E-03 & 4.9174E-03 \\
					90 & 3.3715E-03 & 3.4635E-03 \\
					100 & 4.0715E-03 & 4.1931E-03 \\
					\bottomrule
				\end{tabular}
			\end{table}
			
				\begin{table}
				\centering
				\caption{}
				{\renewcommand{\arraystretch}{3}}
				\begin{tabular}{ccc}
					\toprule
					number & Average-time(s) & Max-time(s) \\
					\midrule
					10 & 8.6974E-04 & 9.0196E-04 \\
					20 & 1.7994E-03 & 1.8218E-03 \\
					30 & 2.8178E-03 & 4.5974E-03 \\
					40 & 3.9610E-03 & 8.0907E-03 \\
					50 & 5.0637E-03 & 8.0278E-03 \\
					60 & 6.1722E-03 & 6.2468E-03 \\
					70 & 7.4670E-03 & 9.6190E-03 \\
					80 & 8.8014E-03 & 1.0946E-02 \\
					90 & 1.0215E-02 & 1.2403E-02 \\
					100 & 1.1952E-02 & 2.6776E-02 \\
					\bottomrule
				\end{tabular}
			\end{table}
			
			In order to deeply analyze the operational performance of the algorithm, we utilize the visualization tool of line charts to visually present the differences between the average and maximum running times of the algorithm. By fitting the data points of the average and maximum running times under different job scales and drawing line charts, we can clearly demonstrate the trends of changes with the job scale and the differences between them.
			
			
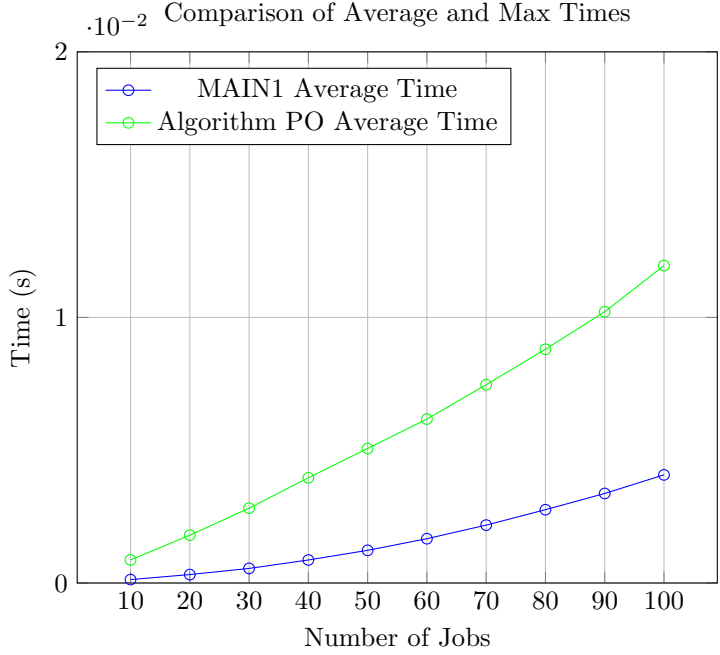
\begin{tikzpicture}
				\begin{axis}[
					xlabel={Number of Jobs},
					ylabel={Time (s)},
					ymin=0,
					ymax=0.02,
					ytick={0,0.01,0.02,0.03},
					xtick={10,20,30,40,50,60,70,80,90,100},
					legend pos=north west,
					grid=both,
					title={Comparison of Average and Max Times },
					]
					
					\addplot[
					color=blue,
					mark=o,
					]
					coordinates {
						(10, 1.2922E-04) (20, 3.1762E-04) (30, 5.4915E-04) (40, 8.6671E-04) (50, 1.2306E-03) (60, 1.6678E-03) (70, 2.1799E-03) (80, 2.7600E-03) (90, 3.3715E-03) (100, 4.0715E-03)
					};
					\addlegendentry{MAIN1 Average Time}

					\addplot[
					color=green,
					mark=o,
					]
					coordinates {
						(10, 8.6974E-04) (20, 1.7994E-03) (30, 2.8178E-03) (40, 3.9610E-03) (50, 5.0637E-03) (60, 6.1722E-03) (70, 7.4670E-03) (80, 8.8014E-03) (90, 1.0215E-02) (100, 1.1952E-02)
					};
					\addlegendentry{Algorithm PO Average Time}

				\end{axis}
			\end{tikzpicture}

		\section{The unbounded case with strict precedence relation}\label{sec4}
		
		In this section we will modify the algorithm developed in the preceding section to solve $1\vert \prec ,s-batch,b\ge n\vert ({{C}_{\max }},{{f}_{\max }})$. As a by-product, the last schedule constructed by the modified algorithm is optimal for $1\vert \prec ,s-batch,b\ge n\vert {{f}_{\max }}$.
		
		Without causing confusion, we can re-use some notations and terminologies defined in the preceding section, such as $S({{B}_{i}})$, $C({{B}_{i}})$, Candidate Set Family (CSF) $F=\{{{\mathcal{J}}_{1}},{{\mathcal{J}}_{2}},\ldots ,{{\mathcal{J}}_{n}}\}$, $\Pi \left( \mathcal{J},F,y \right)$, to name a few. In addition, let $O(j,F)$ denote the ordinal number of job ${{J}_{j}}$ in CSF $F$, i.e., $O(j,F)=i$ means that ${{J}_{j}}\in {{\mathcal{J}}_{i}}$ in $F$. Note that Lemma \ref{start} still holds for $1\vert \prec ,s-batch,b\ge n\vert ({{C}_{\max }},{{f}_{\max }})$.
		
		Precedence relation on $\mathcal{J}$ can be represented as a graph $G=<V,E>$, where $V$ consists of all the jobs in $\mathcal{J}$, and $E$ consists of all the edges $<{{J}_{p}},{{J}_{j}}>$ if ${{J}_{p}}\prec {{J}_{j}}$ and there is no ${{J}_{k}}$ such that ${{J}_{p}}\prec {{J}_{k}}\prec {{J}_{j}}$. To determine the set of direct predecessors for each job, we purposely use the inverse adjacency list representation of $G$. For each job ${{J}_{j}}$, its inverse adjacency list contains all the jobs ${{J}_{p}}$ such that there is an edge $<{{J}_{p}},{{J}_{j}}>\in E$.
		
		There is an natural initial CSF ${{F}^{0}}=\{\mathcal{J}_{1}^{0},\mathcal{J}_{2}^{0},\ldots ,\mathcal{J}_{n}^{0}\}$, where $\mathcal{J}_{i}^{0}$ consists of all the vertices whose out-degrees are zero in $G\backslash \bigcup _{h=i+1}^{n}\mathcal{J}_{h}^{0}$, $i=n,n-1,\ldots ,1$. Clearly, ${{F}^{0}}$ obeys the precedence relation and can be constructed from $G$ in $O({{n}^{2}})$ time.
		
		We first modify Algorithm AUX1 to solve the following auxiliary problem for the unbounded model with strict precedence relation.
		
		{\bf Auxiliary Problem:} Find a schedule $\sigma $ in $\Pi \left( \mathcal{J},F,y \right)$ with minimum makespan.
		\begin{algorithm}
			\caption{(Algorithm Aux2)}\label{<Aux2>}
			\begin{algorithmic}[0]
				\item {\bf Step 1.} Initially, set $m=0$. Let ${{F}^{m}}=\{\mathcal{J}_{1}^{m},\mathcal{J}_{2}^{m},\ldots ,\mathcal{J}_{n}^{m}\}$, where $\mathcal{J}_{i}^{m}={{\mathcal{J}}_{i}}$, $i=1,2,\ldots ,n$. Determine $O(j,{{F}^{m}})$ for job ${{J}_{j}}$, $j=1,2,\ldots ,n$.
				\item {\bf Step 2.} Form the batches of schedule ${{\sigma }^{m}}$: For $i=n,n-1,\ldots ,1$, let batch $B_{_{i}}^{m}=\mathcal{J}_{i}^{m}$.
				\item {\bf Step 3.} Schedule the batches in ${{\sigma }^{m}}$: Let $S(B_{_{1}}^{m})=s(B_{_{1}}^{m})$, $S(B_{_{i}}^{m})=C(B_{_{i-1}}^{m})+s(B_{_{i}}^{m})$, $i=2,\ldots ,n$.
				\item {\bf Step 4.}  Adjust ${{F}^{m}}$: We handle the batches in ${{\sigma }^{m}}$ in decreasing order of their indices and modify the components of ${{F}^{m}}$ accordingly. Batch $B_{_{i}}^{m}$ is handled after all the batches  $B_{_{n}}^{m},B_{_{n-1}}^{m},\ldots ,B_{_{i+1}}^{m}$ have already been handled. It is possible that the ordinal numbers of some jobs in the lower-numbered components of ${{F}^{m}}$ may have been modified when a higher-numbered batch in ${{\sigma }^{m}}$ is handled. Suppose that we are handling batch $B_{_{i}}^{m}$ ($i=n,n-1,\ldots ,1$). Check  the inequality ${{f}_{j}}(C(B_{_{i}}^{m}))<y$ for each job ${{J}_{j}}$ in $B_{_{i}}^{m}$. If the inequality holds, then ${{k}_{1}}=i$; Otherwise let ${{k}_{1}}$ be the largest index such that ${{f}_{j}}(C(B_{_{{{k}_{1}}}}^{m}))<y$. If ${{k}_{1}}$ does not exist (i.e., ${{f}_{j}}(C(B_{_{1}}^{m}))\ge y$), then return $\varnothing $; otherwise update the ordinal number of ${{J}_{j}}$ in ${{F}^{m}}$ to be $k=\min \{{{k}_{1}},O(j,{{F}^{m}})\}$. If $k=i$, then ${{J}_{j}}$ will not be moved (i.e., ${{J}_{j}}$ is still in $\mathcal{J}_{i}^{m}$). Otherwise, delete ${{J}_{j}}$ from $\mathcal{J}_{i}^{m}$  in ${{F}^{m}}$ and insert it in $\mathcal{J}_{k}^{m}$. For the latter case (i.e., $k<i$), if $\mathcal{J}_{i}^{m}$ becomes empty then return $\varnothing $; otherwise we further check the inverse adjacency list of ${{J}_{j}}$. For each job ${{J}_{p}}$ in this list, update its ordinal number in ${{F}^{m}}$ to be $\min \{k-1,O(p,{{F}^{m}})\}$ (but ${{J}_{p}}$ is not moved yet).
				\item {\bf Step 5.} If no adjustment has been done for ${{F}^{m}}$ after Step 4 (i.e., for all pairs $i,j$ the inequality is correct), then return ${{\sigma }^{m}}$. Otherwise, let ${{F}^{m+1}}={{F}^{m}}$ (${{F}^{m}}$ has been adjusted already) and then set $m=m+1$. Go to the next iteration (Step 2).
			\end{algorithmic}
		\end{algorithm}
		
		Step 1 can be implemented in $O(n)$ time. Steps 2, 3 and 5 require $O(n)$ time in each iteration. Since there are at most $O({{n}^{2}})$ iterations, Steps 2, 3 and 5 can be implemented in $O({{n}^{3}})$ time for all iterations.
		
		Let us explain Step 4 a little. If we use adjacency list instead of inverse adjacency list, then Step 4 can be implemented easily in $O({{n}^{2}})$ time for each iteration. Suppose that we are handling batch $B_{_{i}}^{m}$. For each job ${{J}_{j}}$ in $B_{_{i}}^{m}$, we simply update $O(j,{{F}^{m}})$ to be ${k}'=\min \{{{k}_{1}},{{\min }_{<{{J}_{j}},{{J}_{q}}>\in E}}O(q,{{F}^{m}})-1\}$, where ${{k}_{1}}$ is just the one defined in Step 4 of Algorithm AUX2. If ${k}'<i$, delete ${{J}_{j}}$ from $\mathcal{J}_{i}^{m}$  in ${{F}^{m}}$ and insert it in $\mathcal{J}_{{{k}'}}^{m}$. The implementation is easy to understand, since the ordinal number of any job can be modified only once, when the batch containing it is handled. The disadvantage of this implementation is that it requires $O({{n}^{4}})$ time for all iterations.
		
		Instead, we use inverse adjacency list in Step 4, and we distinguish the jobs in $B_{_{i}}^{m}$ between moved and unmoved. For the unmoved jobs in $B_{_{i}}^{m}$, we do nothing again, because they will not affect the ordinal numbers of their predecessors. On the other hand, for each moved job, we update the ordinal numbers of its direct predecessors in ${{F}^{m}}$. Note that any job can be moved at most $n-1$ times. Therefore, the complexity contribution of Step 4 for all iterations is $O({{n}^{3}})$.
		
		In a word, the running time of Algorithm AUX2 is $O({{n}^{3}})$.
		
		Lemma \ref{early}, Lemma \ref{equal1}, and Theorem \ref{last} still hold for Algorithm AUX2.
		
		We now present the main algorithm for constructing the Pareto set $\Omega (\mathcal{J})$ for $1\vert \prec ,s-batch,b\ge n\vert ({{C}_{\max }},{{f}_{\max }})$.
		
		\begin{algorithm}
			\caption{(Algorithm MAIN2)}\label{<MAIN2>}
			\begin{algorithmic}[0]
				\item {\bf Step 1.} Initially, set $u=0$, ${{y}^{u}}=+\infty $ and ${{\sigma }_{u}}=\varnothing $. Let ${{F}^{u}}$ be just the one defined at the beginning of this section (the natural initial CSF ${{F}^{0}}$). Let $\Omega (\mathcal{J})=\varnothing $, $k=0$.
				\item {\bf Step 2.} Run Algorithm AUX2 to get a schedule ${{\sigma }_{u+1}}$ with minimum makespan among all schedules in $\Pi \left( \mathcal{J},{{F}^{u}},{{y}^{u}} \right)$. Let upon the completion of Algorithm AUX2 ${{F}^{u+1}}=\{\mathcal{J}_{1}^{u+1},\mathcal{J}_{2}^{u+1},\ldots ,\mathcal{J}_{n}^{u+1}\}$ be obtained.
				\item {\bf Step 3.} The same as Step 3 of Main1.
				
				\item {\bf Step 4.} The same as Step 4 of Main1.
			\end{algorithmic}
		\end{algorithm}
		
	In this section, we present experimental results of our $1\vert \prec ,s-batch,b\ge n\vert ({{C}_{\max }},{{f}_{\max }})$ algorithm. The algorithm, implemented in PyCharm, was tested on randomly generated instances. We varied key factors. In each experiment, we assign a different number of jobs\((J\in\{10,100\})\), randomly generate the processing time, deadline, and priority relationship for each job. The average and maximum running times of Algorithm  in the paper is showed in Table 2.
		
		\begin{table}[h]
			\centering
			\caption{}
			{\renewcommand{\arraystretch}{3}}
				\begin{tabular}{ccc}
					\toprule
					number &  Average-time(s) & Max-time(s) \\
					\midrule
					10 & 1.2894E - 02 & 3.0905E - 02 \\
					20 & 2.8576E - 02 & 9.443E - 02 \\
					30 & 4.8141E - 02 & 1.8098E - 01 \\
					40 & 6.6872E - 02 & 3.5355E - 01 \\
					50 & 1.6463E - 01 & 8.4042E - 01 \\
					60 & 1.7970E - 01 & 9.982E - 01 \\
					70 & 2.7272E - 01 & 2.0207E + 00 \\
					80 & 3.7795E - 01 & 2.4787E + 00 \\
					90 & 5.2272E - 01 & 3.9724E + 00 \\
					100 & 6.3204E - 01 & 5.0993E + 00 \\
					\bottomrule
				\end{tabular}
			\end{table}

			Lemma \ref{equal2} still holds for Algorithm Main2 and the natural initial CSF ${{F}^{0}}$. We then get:
			
			\begin{theorem} \label{MAIN2}
				Algorithm MAIN2 solves $1\vert \prec ,s-batch,b\ge n\vert ({{C}_{\max }},{{f}_{\max }})$ in $O({{n}^{3}})$ time. Consequently, problem $1\vert \prec ,s-batch,b\ge n\vert {{f}_{\max }}$  can  also be solved in $O({{n}^{3}})$ time.
			\end{theorem}

			\section{Conclusions}\label{sec5}
			
			In this paper we investigated the bicriteria problem of scheduling jobs on a serial-batch machine to minimize makespan and maximum cost simultaneously. We improved the earlier results by presenting $O({{n}^{3}})$-time algorithms for the bounded model without precedence relation and the unbounded model with strict precedence relation. For future research, it is interesting to design an algorithm with better time complexity. It is also interesting to improve the algorithms presented in \citep{Baptiste00} for the problem of scheduling jobs with equal processing times and unequal release dates on a serial-batch machine.

			\backmatter

			\bmhead{Acknowledgments}
			
			This work is supported by Natural Science Foundation of Shandong Province China (No. ZR2020MA030).
			
			\section*{Declarations}
			The authors certify that they have no affiliations with or involvement in any organization or entity with any financial interest or non-financial interest in the subject matter or materials discussed in this manuscript.

		\end{document}